\documentclass[12pt]{article}
\usepackage[margin=1in]{geometry}
\usepackage{amsmath,amssymb,amsthm}
\usepackage{lmodern}
\usepackage[T1]{fontenc}
\usepackage{setspace}
\usepackage{graphicx}
\usepackage{needspace}
\usepackage[dvipsnames]{xcolor}
\usepackage{pgfplots}
\pgfplotsset{compat=1.17}
\usetikzlibrary{calc,positioning}
\definecolor{ublue}{HTML}{2B6CB0}
\definecolor{uyellow}{HTML}{D99A1E}
\definecolor{ugreen}{HTML}{3A9A5B}
\definecolor{ured}{HTML}{C9443A}
\usepackage[font=small,labelfont=bf]{caption}
\usepackage[authoryear,round]{natbib}
\usepackage[colorlinks=true,linkcolor=blue!50!black,citecolor=blue!50!black,urlcolor=blue!50!black]{hyperref}
\hypersetup{pdftitle={The Limits of Price Discrimination When the Market Is Unknown},pdfauthor={Irfan Tekdir},pdfkeywords={Price discrimination, Market segmentation, Information design, Equal-revenue distribution, Robustness}}

\newif\ifanon \anonfalse

\usepackage{etoolbox}
\AtBeginEnvironment{lemma}{\interlinepenalty=10000\predisplaypenalty=10000\postdisplaypenalty=10000}
\AtBeginEnvironment{theorem}{\interlinepenalty=10000\predisplaypenalty=10000\postdisplaypenalty=10000}
\AtBeginEnvironment{proposition}{\interlinepenalty=10000\predisplaypenalty=10000\postdisplaypenalty=10000}
\AtBeginEnvironment{thmA}{\interlinepenalty=10000\predisplaypenalty=10000\postdisplaypenalty=10000}

\newtheorem{theorem}{Theorem}
\newtheorem{proposition}{Proposition}
\newtheorem{lemma}{Lemma}
\newtheorem*{thmA}{Theorem}
\theoremstyle{definition}
\newtheorem{example}{Example}

\newcommand{\conv}{\operatorname{conv}}

\title{The Limits of Price Discrimination\\ When the Market Is Unknown\thanks{I used Anthropic's Claude as a research assistant for editing, typesetting, writing the figure code, and checking derivations and references. All errors are my own.}}
\ifanon
\author{}
\else
\author{Irfan Tekdir\thanks{Department of Economics, The Pennsylvania State University. Email: \href{mailto:tekdirfan@psu.edu}{tekdirfan@psu.edu}.}}
\fi
\date{October 29, 2026}

\begin{document}
\maketitle

\begin{abstract}
\setstretch{1.2}
\noindent \citet{BBM2015} characterize the consumer and producer surplus pairs that segmentation can produce when the distribution of values is known. I fix only the set of possible values and characterize the pairs that some market and some segmentation can produce. They form a convex polygon. Its outer edge is the largest efficient surplus a market can have when its uniform-price profit is capped, net of that profit, and truncated equal-revenue markets attain it. At the values no segmentation is needed; strictly between two values, reaching the frontier requires splitting a market into two extremal markets. With a continuum of values up to $\bar v$, consumer surplus is at most $\pi\ln(\bar v/\pi)$ when producer surplus (the seller's profit) is $\pi$, this bound is reached without segmentation, and no market gives consumers more than $\bar v/e$.

\vspace{1em}
\end{abstract}

\section{Introduction}

Price discrimination runs on information about consumers, and that information is often collected and released by someone other than the seller. A platform decides which consumer characteristics its sellers can condition prices on, a data broker sells customer segments, and a privacy regulator limits what can be used at all. Each of them is choosing how finely the market is segmented, and each faces the same question: how much can information about consumers help or hurt them?

\Needspace{5\baselineskip}
For a monopolist facing a known distribution of values, \citet{BBM2015}, BBM from here on, give a complete answer. Segmentation can achieve every pair of consumer and producer surplus in which consumer surplus is nonnegative, producer surplus is at least the uniform-price profit, and total surplus is at most the efficient surplus.

Using this triangle requires knowing the distribution of values, which is the demand curve, and the parties who judge consumer data rarely know it. A regulator writing a disclosure rule, or a platform deciding what to reveal to sellers, sets one policy for many markets and usually knows at most the range of values consumers place on a good. This paper asks the corresponding robust question. Fixing only the set of possible values, which pairs of consumer and producer surplus can price discrimination produce for some market and some segmentation of it?

The answer is the union of BBM's triangles over all markets, and it is a convex polygon (Theorem~\ref{thm:polygon}, Figure~\ref{fig:polygon}). Its left edge is the producer-surplus axis, where the seller extracts all surplus. Its outer edge is $u=\bar w(\pi)-\pi$, where $\bar w(\pi)$ is the largest efficient surplus any market can have when its uniform-price profit is at most $\pi$. The market that attains $\bar w(\pi)$ has a truncated equal-revenue distribution (Lemma~\ref{lem:wbar}). To keep monopoly profit low while making the market as valuable as possible, every upper tail of the distribution has to be as thick as the profit cap allows, and then the seller is indifferent among all high prices.

Two arguments do the work. The first is a reduction (Lemma~\ref{lem:hull}): once the market may vary, the union of outcomes over all markets and segmentations equals the convex hull of the outcomes that markets produce \emph{without} segmentation. Segmentation adds nothing that mixing unsegmented markets cannot, and information is needed on the frontier only where no single market reaches it, which happens strictly between two values. The second is the profit-cap problem behind $\bar w$, which has a closed-form solution.

The vertices of the polygon sit at the values. At $\pi=v_k$ the frontier market is BBM's extremal market on $\{v_k,\dots,v_K\}$, priced uniformly at $v_k$ with no segmentation at all. Strictly between two values, no unsegmented market reaches the frontier, and reaching it requires splitting a market into two extremal markets (Proposition~\ref{prop:between}). Segmentation matters for the robust frontier only because values are discrete. With a continuum of values (Proposition~\ref{prop:continuum}), every frontier point is reached by an unsegmented equal-revenue market, the frontier becomes $u=\pi\ln(\bar v/\pi)$, and the largest consumer surplus over all markets is $\bar v/e$, reached at profit $\bar v/e$. For the linear demand $D(p)=1-p$, the best segmentation of the linear market gives consumers $\tfrac14$ at profit $\tfrac14$, while another market with the same range of values gives them $\tfrac14\ln4\approx0.35$ at the same profit.

\subsection{Related Literature}

This paper builds directly on BBM and uses their extremal markets. Their approach belongs to the information-design literature that starts with \citet{KamenicaGentzkow2011}, surveyed by \citet{BergemannMorris2019}. \citet{HaghpanahSiegel2023} study when segmentation can make everyone better off when the seller offers several products. All of these papers fix the market. The closest analogue to the question asked here is robust pricing by a seller who does not know the distribution of values \citep{BergemannSchlag2011}. The equal-revenue distribution that attains the frontier is the same family that solves the buyer-optimal learning problem of \citet{RoeslerSzentes2017}, and it also arises when the buyer chooses her value distribution before the seller sets a price \citep{CondorelliSzentes2020}. In \citet{RoeslerSzentes2017} and \citet{CondorelliSzentes2020}, as here, what is best for consumers is whatever removes the seller's reason to prefer a high price. The difference is that in those papers someone chooses the distribution, whereas here the family appears because it carries the most surplus compatible with a given uniform-price profit, and the object of interest is the whole set of surplus pairs rather than one party's optimum. \citet{Tekdir2024} studies a different departure from BBM, in which the market is known but segmentation is costly.

\Needspace{8\baselineskip}
\section{Model}\label{sec:model}

I follow BBM's model and notation. A monopolist with zero marginal cost sells a good to a unit mass of consumers, each of whom demands one unit. There are $K$ possible values $0<v_1<\dots<v_K$. A market is a distribution $x\in X=\Delta(\{v_1,\dots,v_K\})$, where $x_k$ is the share of consumers with value $v_k$. Write $S_k(x)=\sum_{j\ge k}x_j$ for the share of consumers willing to pay $v_k$, so that the uniform price $v_k$ earns $v_kS_k(x)$. Price $v_k$ is optimal for $x$ if $v_kS_k(x)\ge v_iS_i(x)$ for every $i$. For each market let
\[
\pi^*(x)=\max_kv_kS_k(x),\qquad w^*(x)=\sum_{j=1}^Kv_jx_j
\]
be its uniform-price profit and its efficient surplus.

A segmentation of $x^*$ is a distribution $\sigma$ over $X$ with finite support and $\sum_x\sigma(x)x=x^*$; each market in its support is a segment. A pricing rule assigns to each segment $x$ a distribution $\phi(x)$ over prices, and it is optimal if every price charged with positive probability in $x$ is optimal for $x$. Consumers buy whenever the price does not exceed their value. Consumer and producer surplus are
\[
u=\sum_x\sigma(x)\sum_k\phi_k(x)\sum_{j\ge k}(v_j-v_k)x_j,\qquad\pi=\sum_x\sigma(x)\sum_k\phi_k(x)\,v_kS_k(x).
\]
Let $W(x^*)$ be the set of pairs $(u,\pi)$ generated by segmentations of $x^*$ with optimal pricing rules.

\begin{thmA}[\citealp{BBM2015}, Theorem 1]
For every market $x^*$, $W(x^*)=\{(u,\pi):u\ge0,\ \pi\ge\pi^*(x^*),\ u+\pi\le w^*(x^*)\}$.
\end{thmA}

The theorem has two parts. The first is an upper bound: every segmentation with optimal pricing satisfies the three inequalities, so $W(x^*)$ lies inside the triangle. The second is attainment: every point of the triangle is generated by some segmentation, which BBM prove by building segmentations out of extremal markets (defined below). My main results use only the first part, which is easy to see. Consumers buy only at prices weakly below their values, so $u\ge0$, and total surplus counts the values of the consumers who buy, so $u+\pi\le w^*$. If $v_i$ is optimal for $x^*$, then $\pi^*(x)\ge v_iS_i(x)$ in every segment, and since $S_i$ is linear in the market, $\pi=\sum_x\sigma(x)\pi^*(x)\ge v_i\sum_x\sigma(x)S_i(x)=v_iS_i(x^*)=\pi^*(x^*)$.

For a nonempty set of values $S$, let $x^S$ be BBM's \emph{extremal market}: the market supported on $S$ in which the seller is indifferent among all prices in $S$. BBM show that, for $v_i\in S$,
\begin{equation}\label{eq:extremal}
x^S_i=\min S\Big(\frac1{v_i}-\frac1{\mu(v_i,S)}\Big)\ \text{ if }v_i\ne\max S,\qquad x^S_i=\frac{\min S}{\max S}\ \text{ if }v_i=\max S,
\end{equation}
where $\mu(v_i,S)$ is the next element of $S$ above $v_i$. I need~\eqref{eq:extremal} only for upper sets $S=\{v_k,\dots,v_K\}$.

\Needspace{8\baselineskip}
\section{Main Results}\label{sec:main}

\subsection{The Union Is a Convex Hull}

The object of interest is
\[
U=\bigcup_{x\in X}W(x),
\]
the set of pairs that some market produces under some segmentation. For a market $x$, let $O(x)$ be the set of pairs it produces without segmentation, as the seller randomizes over its optimal prices:
\[
O(x)=\Big\{\Big(\textstyle\sum_k\phi_ku_k(x),\ \pi^*(x)\Big):\phi\in\Delta\big(\{k:v_kS_k(x)=\pi^*(x)\}\big)\Big\},
\]
where $u_k(x)=\sum_{j\ge k}(v_j-v_k)x_j$. Every pair in $O(x)$ has producer surplus $\pi^*(x)$, so $O(x)$ is a horizontal segment on the bottom edge of $W(x)$, and a single point when the optimal price is unique.

\begin{lemma}\label{lem:hull}
$U=\conv\big(\bigcup_{x\in X}O(x)\big)$.
\end{lemma}

\begin{proof}
Write $C=\conv\big(\bigcup_{x\in X}O(x)\big)$ for the right-hand side.\par\nopagebreak

\emph{$U\subseteq C$.} Take any $(u,\pi)\in U$. Then $(u,\pi)\in W(\bar x)$ for some market $\bar x$, so it is generated by a segmentation of $\bar x$ with segments $x_1,\dots,x_n$, weights $\sigma_1,\dots,\sigma_n$, and an optimal pricing rule. In segment $x_i$ the rule charges only optimal prices, so the seller earns $\pi^*(x_i)$ there and the segment produces a pair $o_i\in O(x_i)$. Welfare adds up across segments, so $(u,\pi)=\sum_i\sigma_io_i$, a convex combination of points of $\bigcup_xO(x)$. Hence $(u,\pi)\in C$.

\emph{$C\subseteq U$.} Take any point of $C$. By definition of the convex hull it equals $\sum_i\lambda_io_i$ for finitely many markets $x_i$, outcomes $o_i\in O(x_i)$ and weights $\lambda_i\ge0$ summing to one. Each $o_i$ comes from pricing $x_i$ with some distribution over its optimal prices. Since $X$ is convex, $\bar x=\sum_i\lambda_ix_i$ is a market, and splitting $\bar x$ into the segments $x_i$ with weights $\lambda_i$ is a segmentation of it. Pricing each segment as in $o_i$ is an optimal pricing rule, and it generates $\sum_i\lambda_io_i$. If two segments coincide, merge them and price with the weighted mixture of their two price distributions, which is still optimal and, by linearity, gives the same welfare. Hence the point lies in $W(\bar x)\subseteq U$.
\end{proof}

Once the market itself may vary, nothing is lost by looking only at unsegmented outcomes and mixing them. Information is needed to reach a frontier point only if no single market produces it.

\Needspace{8\baselineskip}
\subsection{The Largest Surplus under a Profit Cap}

By Lemma~\ref{lem:hull}, $U$ is the convex hull of the unsegmented outcomes, so it is enough to know which unsegmented outcomes exist. In an unsegmented market the seller earns its uniform-price profit, and consumers get at most what is left of the efficient surplus. This gives an upper bound once we know how much efficient surplus a market can contain when its uniform-price profit is capped. The following function measures exactly that. For $\pi\in[v_1,v_K]$ let
\[
\bar w(\pi)=\max\{w^*(x):x\in X,\ \pi^*(x)\le\pi\},
\]
the largest efficient surplus a market can have when the seller can earn at most $\pi$ with a uniform price. Restricting to $\pi\ge v_1$ loses nothing, since the price $v_1$ sells to every consumer and so $\pi^*(x)\ge v_1$ in every market.

\begin{lemma}\label{lem:wbar}
Let $v_0=0$. For every $\pi\in[v_1,v_K]$,
\[
\bar w(\pi)=\sum_{k=1}^K(v_k-v_{k-1})\min\Big(1,\frac{\pi}{v_k}\Big).
\]
The maximum is attained by the truncated equal-revenue market $x^\pi$ defined by $S_k(x^\pi)=\min(1,\pi/v_k)$. In $x^\pi$, every price $v_k\ge\pi$ earns exactly $\pi$ and every lower price earns less, so $\pi^*(x^\pi)=\pi$. The function $\bar w$ is nondecreasing, concave and piecewise linear, with kinks only at values.
\end{lemma}

\begin{proof}
Since $x_j=S_j-S_{j+1}$ with $S_{K+1}=0$, summation by parts gives $w^*(x)=\sum_k(v_k-v_{k-1})S_k(x)$. Markets correspond one to one with tail sequences $1=S_1\ge S_2\ge\dots\ge S_K\ge0$. The cap $\pi^*(x)\le\pi$ says $v_kS_k\le\pi$ for every $k$, and together with $S_k\le1$ it reads $S_k\le\bar S_k=\min(1,\pi/v_k)$. The requirement $S_1=1$ is compatible with the cap because $\pi\ge v_1$. The objective increases in every $S_k$, since $v_k-v_{k-1}>0$, so it is maximized by setting $S_k=\bar S_k$ for all $k$, provided these bounds form a tail sequence. They do: $\bar S_1=1$, and $\bar S_k$ is nonincreasing because $v_k$ increases. In $x^\pi$, price $v_k$ earns $v_k\bar S_k=\min(v_k,\pi)$, which equals $\pi$ when $v_k\ge\pi$ and is smaller otherwise. Finally, $\bar w$ is nondecreasing because the feasible set grows with $\pi$, and it is a positive combination of the functions $\pi\mapsto\min(1,\pi/v_k)$, each concave and piecewise linear with a single kink at $v_k$.
\end{proof}

The name describes the tail of $x^\pi$. For every value $v_k\ge\pi$, the share of consumers willing to pay $v_k$ is $\pi/v_k$, so the price $v_k$ earns exactly $\pi$: revenue is the same at all these prices, hence ``equal revenue.'' For a value $v_k<\pi$, the formula $\pi/v_k$ would exceed one, and a share cannot, so the tail is cut off at one there; this is the ``truncation.'' A tail of one at $v_k$ means that every consumer is willing to pay $v_k$. So in $x^\pi$ no consumer has a value below the largest value that does not exceed $\pi$, each lower price sells to everyone and earns less than $\pi$, and the seller is indifferent among all prices from $\pi$ up. When $\pi$ is itself one of the values, $x^\pi$ is one of BBM's extremal markets.

\begin{lemma}\label{lem:extremal}
For each $k$, the market $x^{v_k}$ equals BBM's extremal market $x^S$ with $S=\{v_k,\dots,v_K\}$. In it every price in $S$ is optimal, and the price $v_k$ sells to every consumer.
\end{lemma}

\begin{proof}
For $\pi=v_k$ the tails are $S_i=1$ for $i\le k$ and $S_i=v_k/v_i$ for $i\ge k$. Hence $x_i=0$ for $i<k$, $x_i=v_k(1/v_i-1/v_{i+1})$ for $k\le i<K$, and $x_K=v_k/v_K$. With $\min S=v_k$, $\max S=v_K$ and $\mu(v_i,S)=v_{i+1}$, this is~\eqref{eq:extremal}. By Lemma~\ref{lem:wbar}, every price $v_i\ge v_k$ earns $v_k$ and every lower price earns less. No consumer has a value below $v_k$, so $v_k$ sells to everyone.
\end{proof}

\subsection{The Polygon}

\begin{theorem}\label{thm:polygon}
The union is the convex polygon
\[
U=\{(u,\pi):v_1\le\pi\le v_K,\ 0\le u\le\bar w(\pi)-\pi\}
\]
with vertices $(0,v_1)$ and $(\hat u_k,v_k)$, $k=1,\dots,K$, where
\[
\hat u_k=\bar w(v_k)-v_k=v_k\sum_{j=k+1}^K\Big(1-\frac{v_{j-1}}{v_j}\Big),\qquad\hat u_K=0.
\]
\end{theorem}

The vertex $(\hat u_k,v_k)$ is the outcome of the extremal market $x^{\{v_k,\dots,v_K\}}$ under uniform pricing at $v_k$, and $(0,v_1)$ is the outcome of the market in which every consumer has value $v_1$.

\begin{proof}
Write $P$ for the polygon on the right-hand side.\par\nopagebreak

\emph{$U\subseteq P$.} Let $(u,\pi)\in W(x)$. By the bounds of Section~\ref{sec:model}, $u\ge0$ and $\pi\ge\pi^*(x)\ge v_1$, and $\pi\le u+\pi\le w^*(x)\le v_K$. The market $x$ itself satisfies the cap at level $\pi^*(x)$, so $w^*(x)\le\bar w(\pi^*(x))$. Since $\bar w$ is nondecreasing and $\pi^*(x)\le\pi$, we get $u+\pi\le w^*(x)\le\bar w(\pi^*(x))\le\bar w(\pi)$.

\Needspace{11\baselineskip}
\emph{The shape of $P$.} By Lemma~\ref{lem:wbar}, $\pi\mapsto\bar w(\pi)-\pi$ is concave and piecewise linear with kinks only at values, and it is nonnegative because $\bar w(\pi)\ge\pi$. The region between the $\pi$-axis and its graph is therefore the convex hull of $(0,v_1)$, $(0,v_K)$ and the points $(\bar w(v_k)-v_k,v_k)$. The formula for $\hat u_k$ follows from
\[
\bar w(v_k)=\sum_{j\le k}(v_j-v_{j-1})+\sum_{j>k}(v_j-v_{j-1})\frac{v_k}{v_j}=v_k\Big(1+\sum_{j>k}\Big(1-\frac{v_{j-1}}{v_j}\Big)\Big),
\]
and $(\hat u_K,v_K)=(0,v_K)$.

\Needspace{6\baselineskip}
\emph{$P\subseteq U$.} By Lemma~\ref{lem:extremal}, price $v_k$ is optimal in $x^{v_k}$ and sells to everyone, so it yields producer surplus $v_k$ and consumer surplus $w^*(x^{v_k})-v_k=\bar w(v_k)-v_k=\hat u_k$. Hence $(\hat u_k,v_k)\in O(x^{v_k})$. In the market where every consumer has value $v_1$, the only optimal price is $v_1$, which yields $(0,v_1)$. All the points that span $P$ therefore belong to $\bigcup_xO(x)$, and $P\subseteq U$ by Lemma~\ref{lem:hull}.
\end{proof}

\Needspace{6\baselineskip}
With the second part of BBM's theorem, attainment, the union can also be read as the set swept out by the bottom edges of the triangles $W(x^\pi)$: the bottom edge of $W(x^\pi)$ is the slice of $U$ at height $\pi$, so every point of the union is reached by some segmentation of $x^\pi$ that leaves the seller exactly its uniform-price profit. Between two values, points on the frontier are mixtures of the two adjacent vertices, so reaching them requires segmentation (Proposition~\ref{prop:between}).

\begin{figure}[!htb]
\centering
\begin{tikzpicture}
\begin{axis}[
  name=pa, width=8.2cm, height=6.6cm,
  xmin=-0.04, xmax=1.12, ymin=0.7, ymax=3.15,
  xlabel={Consumer surplus $u$}, ylabel={Producer surplus $\pi$},
  title={(a) Union over markets, $v=(1,2,3)$}, title style={font=\small, name=titleA},
  tick label style={font=\small},
  legend columns=3,
  legend style={at={(1.02,-0.25)}, anchor=north, draw=none, font=\footnotesize,
                /tikz/every even column/.append style={column sep=8pt}},
  legend cell align=left,
]
\addplot[area legend, fill=gray!22, draw=none] coordinates {(0,1) (0.83333,1) (0.66667,2) (0,3)} -- cycle;
\addlegendentry{Union $U$ over all markets}
\addplot[black, thick] coordinates {(0.83333,1) (0.66667,2) (0,3)};
\addlegendentry{Frontier $u=\bar w(\pi)-\pi$}
\addplot[uyellow, dashed, thick] coordinates {(0,1) (0,1.83333) (0.83333,1) (0,1)};
\addlegendentry{$W(x^S)$, $S=\{1,2,3\}$}
\addplot[ugreen, dashed, thick] coordinates {(0,2) (0,2.66667) (0.66667,2) (0,2)};
\addlegendentry{$W(x^S)$, $S=\{2,3\}$}
\addplot[area legend, fill=ublue, fill opacity=0.22, draw=ublue, draw opacity=0.6]
  coordinates {(0,1.33333) (0,2) (0.66667,1.33333)} -- cycle;
\addlegendentry{$W(x^*)$, $x^*=(\frac13,\frac13,\frac13)$}
\addlegendimage{black, line width=1.4pt}
\addlegendentry{$x^{\pi}$, $\pi\in[v_1,v_3]$}
\addplot[only marks, mark=*, mark size=1.5pt, black, forget plot]
  coordinates {(0,1) (0.83333,1) (0.66667,2) (0,3)};
\node[anchor=west, font=\footnotesize] at (axis cs:0.02,3.0) {$(0,v_3)$};
\node[anchor=west, font=\footnotesize] at (axis cs:0.68,2.02) {$(\hat u_2,v_2)$};
\node[anchor=west, font=\footnotesize] at (axis cs:0.847,1.02) {$(\hat u_1,v_1)$};
\node[anchor=north west, font=\footnotesize] at (axis cs:0.013,0.985) {$(0,v_1)$};
\addplot[only marks, mark=*, mark size=1.6pt, ublue, forget plot] coordinates {(0.33333,1.33333)};
\node[ublue, font=\footnotesize] at (axis cs:0.33,1.18) {uniform pricing};
\end{axis}
\begin{scope}[shift={($(pa.south east)+(1.2cm,0)$)}, x=5.1cm, y=5.1cm]
  \coordinate (bmid) at (0.5,0);
  \node[font=\small, anchor=base] at (bmid |- titleA.base) {(b) Markets attaining the frontier};
  \fill[uyellow, opacity=0.3] (0,0) -- (0.5,0) -- (0.3333,0.2887) -- (0.1667,0.2887) -- cycle;
  \fill[ugreen,  opacity=0.3] (0.5,0) -- (1,0) -- (0.6667,0.5774) -- (0.3333,0.2887) -- cycle;
  \fill[ured,    opacity=0.3] (0.6667,0.5774) -- (0.5,0.8660) -- (0.1667,0.2887) -- (0.3333,0.2887) -- cycle;
  \draw[gray, thin] (0.5,0) -- (0.3333,0.2887) -- (0.1667,0.2887);
  \draw[gray, thin] (0.3333,0.2887) -- (0.6667,0.5774);
  \draw[black] (0,0) -- (1,0) -- (0.5,0.8660) -- cycle;
  \node at (0.24,0.12) {$X_1$};
  \node at (0.68,0.2) {$X_2$};
  \node at (0.42,0.55) {$X_3$};
  \draw[black, line width=1.4pt] (0.3333,0.2887) -- (0.6667,0.5774) -- (0.5,0.8660);
  \fill (0.3333,0.2887) circle (1.5pt) (0.6667,0.5774) circle (1.5pt) (0.5,0.8660) circle (1.5pt);
  \draw (0.3333,0.2887) -- (0.13,0.40) node[anchor=east, font=\footnotesize] {$x^{\{1,2,3\}}$};
  \node[anchor=west, font=\footnotesize] at (0.68,0.59) {$x^{\{2,3\}}$};
  \fill[ublue] (0.5,0.2887) circle (1.6pt);
  \node[ublue, anchor=west, font=\footnotesize] at (0.52,0.28) {$x^*$};
  \node[gray, anchor=north east, font=\footnotesize] at (0,0) {$x^{\{1\}}$};
  \node[gray, anchor=north west, font=\footnotesize] at (1,0) {$x^{\{2\}}$};
  \node[gray, anchor=south, font=\footnotesize] at (0.5,0.87) {$x^{\{3\}}$};
\end{scope}
\end{tikzpicture}
\caption{Values $v=(1,2,3)$. (a) The union $U$ (grey) has vertices $(0,1)$, $(\hat u_1,1)=(\frac56,1)$, $(\hat u_2,2)=(\frac23,2)$ and $(0,3)$. Dashed: BBM's triangles for the extremal markets on $\{1,2,3\}$ and $\{2,3\}$, whose consumer-optimal corners are vertices of $U$. Blue: the triangle of $x^*=(\frac13,\frac13,\frac13)$, with uniform pricing marked. (b) The simplex of markets with BBM's regions $X_k$, in which price $v_k$ is optimal. The black path is the family of truncated equal-revenue markets $x^\pi$, $\pi\in[1,3]$, which passes through the extremal markets that generate the vertices.}
\label{fig:polygon}
\end{figure}
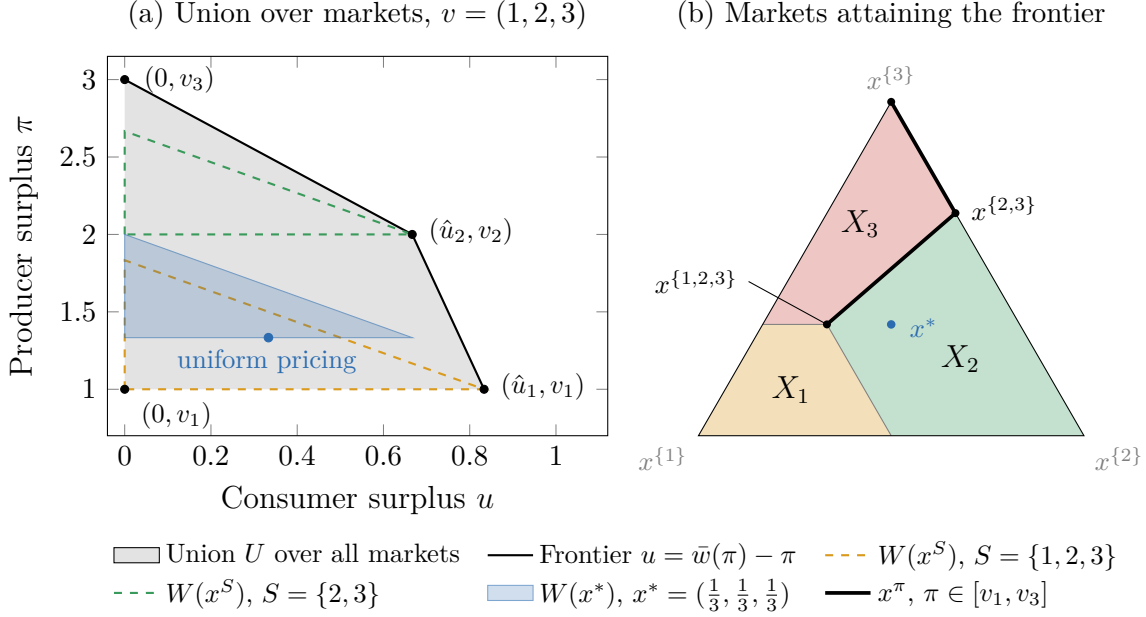

\Needspace{6\baselineskip}
\begin{example}\label{ex:123}
Figure~\ref{fig:polygon}(a) draws the union for BBM's leading example, with values $1$, $2$ and $3$. Here $\hat u_1=1\cdot(\tfrac12+\tfrac13)=\tfrac56$ and $\hat u_2=2\cdot\tfrac13=\tfrac23$. The uniform market $x^*=(\tfrac13,\tfrac13,\tfrac13)$ has $\pi^*=\tfrac43$ and $w^*=2$, so its consumer-optimal outcome gives consumers $\tfrac23$. The frontier at the same profit is $\bar w(\tfrac43)-\tfrac43=\tfrac79$. Some other market with the same three values therefore lets consumers gain $\tfrac19$ more while the seller earns the same.
\end{example}

\Needspace{8\baselineskip}
\subsection{Where Segmentation Is Needed}

Every vertex of the union comes from an unsegmented market. The next result shows that this is special to the vertices.

\begin{proposition}\label{prop:between}
Let $\pi\in[v_1,v_K]$. The frontier point $(\bar w(\pi)-\pi,\pi)$ is generated by some market without segmentation if and only if $\pi$ is a value. If $\pi=\lambda v_k+(1-\lambda)v_{k+1}$ with $\lambda\in(0,1)$, the point is generated by the market $x^\pi=\lambda x^{v_k}+(1-\lambda)x^{v_{k+1}}$, split into these two components and each priced uniformly at its lowest value.
\end{proposition}

\Needspace{8\baselineskip}
\begin{proof}
If $\pi=v_k$, Theorem~\ref{thm:polygon} gives the market. Conversely, suppose a market $x$ generates the frontier point without segmentation, under some distribution over its optimal prices. Every optimal price earns $\pi^*(x)$, so $\pi=\pi^*(x)$, and $\bar w(\pi)=u+\pi\le w^*(x)\le\bar w(\pi^*(x))=\bar w(\pi)$. Hence $u+\pi=w^*(x)$: every consumer buys whatever price is drawn. If price $v_k$ is charged with positive probability, no consumer can have a value below $v_k$, so $S_k(x)=1$ and $\pi=v_kS_k(x)=v_k$. Hence $\pi$ is a value.

For the construction, the tails $S_i(x^\pi)=\min(1,\pi/v_i)$ are linear in $\pi$ on $[v_k,v_{k+1}]$: they equal one for $i\le k$ and $\pi/v_i$ for $i>k$. Hence $x^\pi=\lambda x^{v_k}+(1-\lambda)x^{v_{k+1}}$. Since $\bar w(\pi)-\pi$ is also linear on $[v_k,v_{k+1}]$, $(\bar w(\pi)-\pi,\pi)=\lambda(\hat u_k,v_k)+(1-\lambda)(\hat u_{k+1},v_{k+1})$. Splitting $x^\pi$ into $x^{v_k}$ with weight $\lambda$ and $x^{v_{k+1}}$ with weight $1-\lambda$, and pricing each at its lowest value, generates this average by Theorem~\ref{thm:polygon}.
\end{proof}

Between two values, the frontier market $x^\pi$ would be priced inefficiently if it were left whole. Segmentation serves its low-value consumers at a low price and its high-value consumers at a higher one, so everyone buys. The information involved is coarse: the seller learns only which of two extremal markets a consumer belongs to.

\Needspace{16\baselineskip}
\section{A Continuum of Values}

Suppose now that values lie in an interval $[\underline v,\bar v]$ with $0\le\underline v<\bar v$, and that a market is any probability distribution $F$ on this interval. I describe a market by its tail
\[
S(t)=\Pr(v\ge t),\qquad v\sim F,
\]
the share of consumers whose value is at least $t$. The function $S$ is nonincreasing and left-continuous, with $S(t)=1$ for $t\le\underline v$ and $S(t)=0$ for $t>\bar v$. A uniform price $t$ sells to the share $S(t)$ and earns $tS(t)$. This revenue is upper semicontinuous, so it attains its maximum $\pi^*(F)$ on $[\underline v,\bar v]$.

Efficient surplus is the average value, and it equals the area under the tail:
\[
w^*(F)=\int v\,dF(v)=\int_0^{\bar v}S(t)\,dt.
\]
This is the familiar fact that the mean of a nonnegative random variable is the integral of its tail, or, in demand terms, that efficient surplus is the area under the demand curve. To see it, integrate over the region $\{(v,t):0\le t\le v\le\bar v\}$ with weight $dF(v)$ in $v$ and length $dt$ in $t$, and change the order of integration (Tonelli). Integrating first over $t$ gives the length $v$, hence $\int v\,dF(v)$; integrating first over $v$ gives the share $S(t)$ of consumers with $v\ge t$, hence $\int_0^{\bar v}S(t)\,dt$.\footnote{For example, with values $1,2,3$ in equal shares, integrating first over $t$ gives $(1+2+3)/3=2$, and integrating first over $v$ gives $1+\tfrac23+\tfrac13=2$, since the tail is $1,\tfrac23,\tfrac13$ on the three unit intervals.} This is the continuous version of the formula $w^*(x)=\sum_k(v_k-v_{k-1})S_k(x)$ used in Lemma~\ref{lem:wbar}.

Segmentations are defined as before: a segmentation of $F$ is a finite collection of markets $F_1,\dots,F_n$ with weights $\sigma_i$ such that $\sum_i\sigma_iF_i=F$, and an optimal pricing rule charges each segment a distribution over its optimal prices, which exist because $tS(t)$ is upper semicontinuous. $O(F)$ is the set of outcomes of pricing $F$ itself in this way, and $W(F)$ and $U$ are defined as before. The three bounds of Section~\ref{sec:model} and Lemma~\ref{lem:hull} hold by the same arguments, since these use only that $S$ and the mean value are linear in the market and that the set of markets is convex. None of the arguments below uses the attainment part of BBM's theorem, whose extension to continuous values requires a limiting argument.

The argument follows the discrete case. A cap $\pi$ on profit is a cap on the tail: if no price earns more than $\pi$, then $tS(t)\le\pi$, so $S(t)\le\pi/t$ at every $t$. Since also $S(t)\le1$, the tail lies below the curve $\min(1,\pi/t)$ at every level $t$. Efficient surplus is the area under the tail, so it is at most the area under this curve on $[0,\bar v]$, which equals $\pi+\pi\ln(\bar v/\pi)$ (computed below). The bound is attained exactly when the tail equals the curve on $[0,\bar v]$, that is, by the equal-revenue market
\[
S^\pi(t)=\begin{cases}1,& t\le\pi,\\ \pi/t,& \pi<t\le\bar v,\\ 0,& t>\bar v.\end{cases}
\]
\Needspace{7\baselineskip}\noindent No consumer in it has a value below $\pi$. Values above $\pi$ have density $\pi/t^2$ on $(\pi,\bar v)$, and the remaining share $\pi/\bar v$ sits at $\bar v$; the two parts add up to $(1-\pi/\bar v)+\pi/\bar v=1$. Every price $t\in[\pi,\bar v]$ earns $t\cdot\pi/t=\pi$, and a price $t<\pi$ earns only $t$. The seller is therefore indifferent among all prices from $\pi$ to $\bar v$. Its efficient surplus is
\[
\int_0^{\bar v}\min\Big(1,\frac{\pi}{t}\Big)dt=\int_0^{\pi}dt+\int_\pi^{\bar v}\frac{\pi}{t}\,dt=\pi+\pi\ln\frac{\bar v}{\pi}.
\]
Subtracting the profit $\pi$ leaves
\[
g(\pi)=\pi\ln\frac{\bar v}{\pi}
\]
for consumers, and this is the frontier.

\begin{proposition}\label{prop:continuum}
With values in $[\underline v,\bar v]$,
\[
U=\Big\{(u,\pi):\underline v\le\pi\le\bar v,\ 0\le u\le\pi\ln\frac{\bar v}{\pi}\Big\}.
\]
Every point $(\pi\ln(\bar v/\pi),\pi)$ of the frontier is generated without segmentation by the equal-revenue market $S^\pi$, priced at $\pi$. If $\bar v/e\ge\underline v$, the largest consumer surplus over all markets is $\bar v/e$, attained at producer surplus $\bar v/e$. Otherwise it is $\underline v\ln(\bar v/\underline v)$, attained at $\underline v$.
\end{proposition}

\begin{proof}
Write $P$ for the set on the right-hand side, and set $g(0)=0$, the limit of $g$ at zero.\par\nopagebreak

\emph{Step 1: a cap on profit caps efficient surplus.} Take any market $F$ and write $\pi^*=\pi^*(F)$. Every price earns at most $\pi^*$, so $tS(t)\le\pi^*$ for all $t>0$. Combined with $S(t)\le1$, this gives $S(t)\le\min(1,\pi^*/t)$, and integrating,
\[
w^*(F)=\int_0^{\bar v}S(t)\,dt\le\int_0^{\bar v}\min\Big(1,\frac{\pi^*}{t}\Big)dt=\pi^*+g(\pi^*).
\]
The last equality is the computation above, which applies because $\pi^*\le\bar v$.

\emph{Step 2: $U\subseteq P$.} Let $(u,\pi)\in W(F)$ for some market $F$. By BBM's bounds, $u\ge0$, $\pi\ge\pi^*$ and $u+\pi\le w^*(F)$. The function $h(p)=p+g(p)$ has derivative $h'(p)=\ln(\bar v/p)\ge0$ on $(0,\bar v]$, so it is nondecreasing there. Producer surplus is at most total surplus, which is at most $\bar v$, so $\pi^*\le\pi\le\bar v$ and therefore $h(\pi^*)\le h(\pi)$. Combining this with Step 1,
\[
u+\pi\le w^*(F)\le h(\pi^*)\le h(\pi)=\pi+g(\pi),
\]
that is, $u\le g(\pi)$. Finally $\pi\ge\pi^*\ge\underline v$, because the price $\underline v$ sells to everyone and earns $\underline v$. Hence $(u,\pi)\in P$.

\Needspace{7\baselineskip}
\emph{Step 3: $P\subseteq U$.} Fix $\pi\in[\underline v,\bar v]$. The market $S^\pi$ is supported on $[\pi,\bar v]$, which lies inside $[\underline v,\bar v]$, so it is an admissible market. As shown above, $\pi$ is an optimal price in it, and at that price everyone buys. Producer surplus is $\pi$ and consumer surplus is $w^*-\pi=g(\pi)$, so $(g(\pi),\pi)\in O(S^\pi)$. In the market where every consumer has value $\pi$, the only optimal price is $\pi$, which gives $(0,\pi)$. Every point $(u,\pi)$ with $0\le u\le g(\pi)$ is a convex combination of these two outcomes, so it lies in $U$ by Lemma~\ref{lem:hull}.

\emph{Step 4: the maximum.} The derivatives are $g'(\pi)=\ln(\bar v/\pi)-1$ and $g''(\pi)=-1/\pi<0$, so $g$ is strictly concave and peaks where $\ln(\bar v/\pi)=1$, at $\pi=\bar v/e$. There $g(\bar v/e)=(\bar v/e)\cdot1=\bar v/e$. If $\bar v/e\ge\underline v$, this peak lies in the admissible range $[\underline v,\bar v]$. Otherwise $g'<0$ on the whole range, and the maximum is at the lowest admissible profit, $g(\underline v)=\underline v\ln(\bar v/\underline v)$. Figure~\ref{fig:continuum} in the example below shows the first case.
\end{proof}

\Needspace{8\baselineskip}
The contrast with the discrete case lies in where the seller prices. With finitely many values and $v_k<\pi<v_{k+1}$, the frontier market $x^\pi$ keeps all its lowest consumers at $v_k$, below $\pi$. Charging $v_k$ earns only $v_k<\pi$, so the seller prices at $v_{k+1}$ and the consumers at $v_k$ do not buy. Segmentation is needed to serve them (Proposition~\ref{prop:between}). With a continuum, the equal-revenue market can start exactly at $\pi$. Its lowest value is then itself an optimal price, everyone buys, and the frontier is reached without any information.

\paragraph{Example: linear demand.} The linear demand $D(p)=1-p$ comes from a unit mass of consumers with values uniform on $[0,1]$, so $\underline v=0$ and $\bar v=1$. By Proposition~\ref{prop:continuum}, the union is
\[
U=\{(u,\pi):0\le\pi\le1,\ 0\le u\le\pi\ln(1/\pi)\}.
\]
\begin{figure}[!htb]
\centering
\begin{tikzpicture}
\begin{axis}[
  width=8.6cm, height=7cm,
  xmin=-0.012, xmax=0.56, ymin=0, ymax=1.04,
  xlabel={Consumer surplus $u$}, ylabel={Producer surplus $\pi$},
  xtick={0,0.1,...,0.5}, ytick={0,0.2,...,1},
  tick label style={font=\small},
  legend columns=2,
  legend style={at={(0.5,-0.2)}, anchor=north, draw=none, font=\footnotesize,
                /tikz/every even column/.append style={column sep=10pt}},
  legend cell align=left,
]
\addplot[fill=gray!22, draw=none, area legend, domain=0.0001:1, samples=200, variable=\t]
  ({\t*ln(1/\t)},{\t}) -- (axis cs:0,0) -- cycle;
\addlegendentry{Union $U$}
\addplot[black, thick, domain=0.0001:1, samples=200, variable=\t]
  ({\t*ln(1/\t)},{\t});
\addlegendentry{Frontier $u=\pi\ln(1/\pi)$}
\addplot[uyellow, dashed, thick] coordinates {(0,0.05) (0,0.19979) (0.14979,0.05) (0,0.05)};
\addlegendentry{$W(x^{\pi})$, $\pi=0.05$}
\addplot[ugreen, dashed, thick] coordinates {(0,0.6) (0,0.90649) (0.30649,0.6) (0,0.6)};
\addlegendentry{$W(x^{\pi})$, $\pi=0.6$}
\addplot[area legend, fill=ublue, fill opacity=0.22, draw=ublue, draw opacity=0.6] coordinates {(0,0.25) (0,0.5) (0.25,0.25)} -- cycle;
\addlegendentry{$W$ of the linear market}
\addplot[only marks, mark=*, mark size=1.4pt, uyellow, forget plot] coordinates {(0.14979,0.05)};
\addplot[only marks, mark=*, mark size=1.4pt, ugreen, forget plot] coordinates {(0.30649,0.6)};
\addplot[only marks, mark=*, mark size=1.6pt, ublue, forget plot] coordinates {(0.125,0.25)};
\draw[ublue, thin] (axis cs:0.125,0.25) -- (axis cs:0.15,0.205);
\node[ublue, font=\footnotesize] at (axis cs:0.17,0.18) {uniform pricing};
\addplot[only marks, mark=*, mark size=1.6pt, black, forget plot] coordinates {(0.36788,0.36788)};
\node[anchor=west, font=\footnotesize] at (axis cs:0.375,0.36788) {$(1/e,\,1/e)$};
\end{axis}
\end{tikzpicture}
\caption{Values uniform on $[0,1]$, the unit-mass population behind $D(p)=1-p$. The union $U$ (grey) with frontier $u=\pi\ln(1/\pi)$. Dashed: BBM's triangles of the equal-revenue markets with $\pi=0.05$ and $\pi=0.6$, whose bottom-right corners (dots) lie on the frontier. Blue: BBM's triangle of the linear market, the limits of price discrimination when the market is known, with uniform pricing marked; its consumer-optimal corner $(\tfrac14,\tfrac14)$ lies inside the union.}
\label{fig:continuum}
\end{figure}
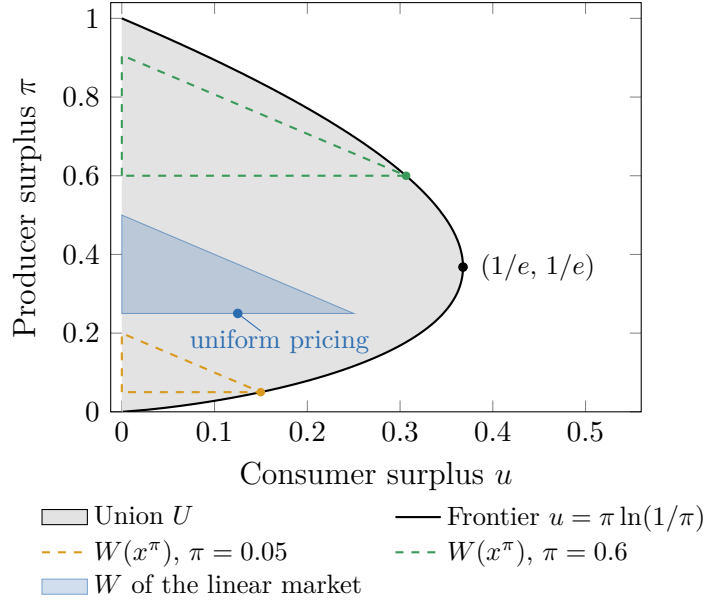

The two shaded objects in Figure~\ref{fig:continuum} answer different questions. The blue triangle is BBM's triangle $W(F)$ for the linear market alone: by their theorem, it is the set of all surplus pairs that segmenting this one market can produce, that is, the limits of price discrimination when the market is known to be linear. (With a continuum of values and finitely many segments, points on its edges may only be approached by finer and finer segmentations.) The grey union collects these triangles over every market with values in $[0,1]$, so it is what remains possible when only the range of values is known. The linear triangle lies strictly inside the union.

Consider the linear market itself, drawn in Figure~\ref{fig:demand}(a). Revenue $p(1-p)$ is largest at $p=\tfrac12$, so the uniform price is $\tfrac12$ and profit is $\pi^*=\tfrac14$. Efficient surplus is $w^*=\int_0^1(1-p)\,dp=\tfrac12$, and consumers who buy at price $\tfrac12$ get $\int_{1/2}^1(1-p)\,dp=\tfrac18$. By BBM's theorem, the best that segmenting this market can do for consumers is $w^*-\pi^*=\tfrac14$.

Now hold profit at $\tfrac14$ and let the market change. The frontier market, drawn in Figure~\ref{fig:demand}(b), has tail $S(t)=\min(1,1/(4t))$ on $[0,1]$. Every consumer has a value of at least $\tfrac14$, and the share $\tfrac14$ that remains at $t=1$ forms an atom at the top value. Its efficient surplus is $\int_0^1\min(1,1/(4t))\,dt=\tfrac14+\tfrac14\ln4$, every price from $\tfrac14$ to $1$ earns exactly $\tfrac14$, and charging $\tfrac14$ sells to everyone. Consumers get $\tfrac14\ln4\approx0.347$, about $39\%$ more than the best segmentation of the linear market. Compared with the linear market, this market has more consumers at every positive price other than $\tfrac12$: below $\tfrac14$ everyone buys, and for $p\ge\tfrac14$ the difference in demand is $1/(4p)-(1-p)=(2p-1)^2/(4p)$. In Figure~\ref{fig:demand}(b) the dashed linear demand lies below the frontier demand at every positive price except $\tfrac12$, where the two curves touch. \citet[Figure~1]{CondorelliSzentes2020} draw the same comparison, starting from the uniform distribution, to show how a buyer who chooses her own distribution can do better than the uniform one. Over all profit levels, $\pi\ln(1/\pi)$ peaks at $\pi=1/e\approx0.37$, where consumer surplus also equals $1/e$. Figure~\ref{fig:continuum} shows the union.

\begin{figure}[!htb]
\centering
\begin{tikzpicture}
\pgfplotsset{dem/.style={
  width=6.4cm, height=5.4cm, xmin=0, xmax=1.08, ymin=0, ymax=1.15,
  axis lines=left, xlabel={Quantity}, ylabel={Price $p$},
  xtick={0,0.25,0.5,0.75,1}, xticklabels={$0$,$\tfrac14$,$\tfrac12$,$\tfrac34$,$1$},
  ytick={0,0.25,0.5,0.75,1}, yticklabels={$0$,$\tfrac14$,$\tfrac12$,$\tfrac34$,$1$},
  tick label style={font=\footnotesize}, label style={font=\small},
  title style={font=\small}, clip=false}}
\begin{axis}[dem, name=da, title={(a) Linear demand $1-p$}]
\addplot[fill=ured, fill opacity=0.22, draw=none] coordinates {(0,0) (0.5,0) (0.5,0.5) (0,0.5)} -- cycle;
\addplot[fill=ugreen, fill opacity=0.3, draw=none] coordinates {(0,0.5) (0,1) (0.5,0.5)} -- cycle;
\addplot[ublue, very thick, domain=0:1] {1-x};
\fill (axis cs:0.5,0.5) circle (1.4pt);
\node[font=\footnotesize, align=center] at (axis cs:0.25,0.25) {seller\\$\tfrac14$};
\node[font=\scriptsize, align=center] at (axis cs:0.165,0.585) {consumers\\$\tfrac18$};
\node[font=\footnotesize, anchor=west] at (axis cs:0.52,0.55) {price $\tfrac12$};
\end{axis}
\begin{axis}[dem, at={($(da.east)+(1.4cm,0)$)}, anchor=west,
  title={(b) Frontier demand $\min(1,1/(4p))$}]
\addplot[fill=ured, fill opacity=0.22, draw=none] coordinates {(0,0) (1,0) (1,0.25) (0,0.25)} -- cycle;
\addplot[fill=ugreen, fill opacity=0.3, draw=none, domain=0.25:1, samples=80, variable=\t]
  ({1/(4*\t)},{\t}) -- (axis cs:0,1) -- (axis cs:0,0.25) -- cycle;
\addplot[ublue, thick, dashed, domain=0:1] {1-x};
\addplot[black, very thick, domain=0.25:1, samples=80, variable=\t] ({1/(4*\t)},{\t});
\draw[black, very thick] (axis cs:0,1) -- (axis cs:0.25,1);
\draw[black, very thick] (axis cs:1,0.25) -- (axis cs:1,0);
\draw[gray, densely dotted] (axis cs:0,0.5) -- (axis cs:0.5,0.5) -- (axis cs:0.5,0);
\draw[gray, densely dotted] (axis cs:0,0.75) -- (axis cs:0.33333,0.75) -- (axis cs:0.33333,0);
\draw[gray, densely dotted] (axis cs:0.25,1) -- (axis cs:0.25,0);
\fill (axis cs:0.5,0.5) circle (1.3pt) (axis cs:0.33333,0.75) circle (1.3pt) (axis cs:0.25,1) circle (1.3pt);
\node[font=\footnotesize, align=center] at (axis cs:0.5,0.125) {seller $\tfrac14\cdot 1=\tfrac14$};
\node[font=\footnotesize, align=center, anchor=west] (cb) at (axis cs:0.6,0.66) {consumers\\$\tfrac14\ln 4\approx 0.35$};
\draw[-latex, thin] (cb.south west) to[bend left=15] (axis cs:0.36,0.42);
\node[font=\footnotesize, anchor=south west] at (axis cs:0.27,1.0) {atom of $\tfrac14$ at value $1$};
\end{axis}
\end{tikzpicture}
\caption{Two markets with values in $[0,1]$ and the same uniform-price profit $\tfrac14$. Red: the seller's revenue; green: consumer surplus. (a) The linear market: the seller charges $\tfrac12$, half the consumers buy, and consumers get $\tfrac18$. (b) The frontier market with demand $\min(1,1/(4p))$: every price from $\tfrac14$ to $1$ earns $p\cdot\frac{1}{4p}=\tfrac14$ (dotted rectangles), so the seller can charge $\tfrac14$ and sell to everyone, and consumers get $\tfrac14\ln4\approx0.35$. Dashed: the linear demand, which meets the frontier demand only at $p=\tfrac12$ and at $p=0$. A similar picture appears as Figure~1 in \citet{CondorelliSzentes2020}.}
\label{fig:demand}
\end{figure}
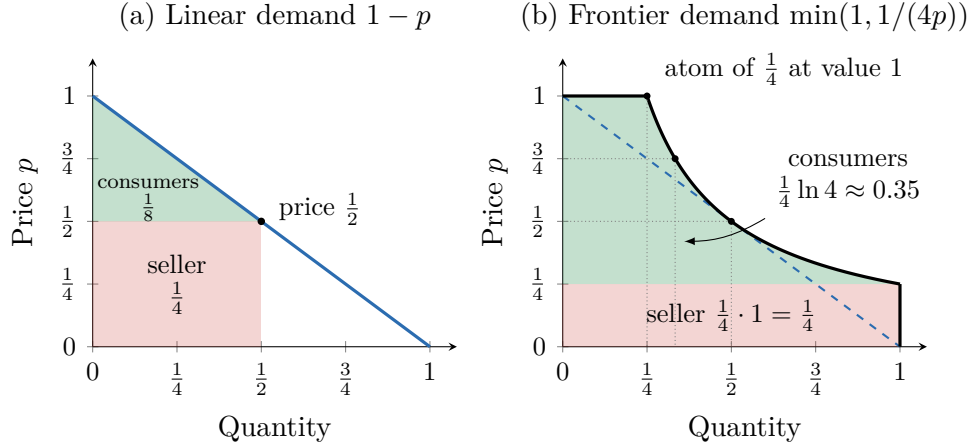

\Needspace{8\baselineskip}
\section{Concluding Remarks}

When the market is not known, the relevant object is the union of BBM's triangles, and its structure is simple. Its outer edge is the most valuable market consistent with a cap on monopoly profit, net of that profit, and that market leaves the seller indifferent among all high prices. With finitely many values the union is a polygon with one vertex per value, and each vertex is easy to compute from the values.

Information plays a small role. When the profit level equals one of the values, the frontier market needs no segmentation. Between two values it has to be split, but two segments are enough, each the frontier market of an adjacent value. With a continuum of values even this disappears, and no point on the frontier requires segmentation. Consumer surplus is then at most $\bar v/e$, whatever the market and whatever the information. A bound on consumer surplus that holds for every market is therefore mostly about which populations of buyers are possible, and much less about what data the seller holds.

\end{document}